\documentclass[preprint,12pt,english]{elsarticle}
\usepackage{amsmath}
\usepackage{amssymb}
\usepackage{lipsum}
\usepackage[sans]{dsfont}
\usepackage{multirow}
\usepackage{afterpage}
\usepackage{amsfonts}
\usepackage{bm}%
\usepackage{babel}
\usepackage{hyperref}
\usepackage{graphicx}
\usepackage{color}
\usepackage{amsthm}
\def\fnum@figure{\figurename\thefigure}
\renewcommand{\figurename}{Fig.}

\newtheorem{theorem}{Theorem}

\journal{Physics Letter A}
\title{Weak measurement as a tool for studying coherence and quantum correlations in bipartite systems}
\begin{document}
  \author{Indrajith V S$^{* @}$ , R. Muthuganesan$^\ddag$ R. Sankaranarayanan$^*$ }
\address{$^*$ Department of Physics, National Institute of Technology\\ Tiruchirappalli-620 015, Tamil Nadu, India.}
\address{$@$ Department of Physics, Indian Institute of Science Education and Research\\ Mohali-140306, Punjab, India.}
\address{$^\ddag$ Department of Physics, Faculty of Nuclear Sciences and Physical Engineering,
Czech Technical University in Prague, Brehova 7, 115 19 Praha 1-Stare Mesto, Czech Republic}
\begin{abstract}
In this article, we study quantum coherence of bipartite state from the perspective of weak measurement, which generalizes the notion of coherence relative to measurement. The is being illustrated by computing coherence for the well-known Bell diagonal and Wener states. We have also extended our investigation on quantum correlation measure and uncertainty relation in the weak measurement regime.  
\end{abstract}
\maketitle
\section{Introduction}
Coherence is an important aspect of quantum regime that arises from the property of superposition \cite{coh}. It is also considered as a physical resource like entanglement, and so useful for various information processing tasks \cite{qkd}-\cite{quant_bio4}. Unlike entanglement, quantum coherence is a basis dependent quantity that can be defined for single and multipartite systems as well. Coherence plays an indispensable role in various fields such as quantum thermodynamics \cite{quant_thrmo_1, quant_thrmo_2, quant_thrmo_3}, communication \cite{quat_comm_1, quat_comm_2, quat_comm_3}, metrology \cite{coh_met1, coh_met2}, quantum biology \cite{quant_bio1, quant_bio2, quant_bio3, quant_bio4} etc. Admired by the seminal framework of Baumgratz \cite{coh}, many such coherence measures have been proposed based on both entropy and distance measures \cite{cohe_hellin} - \cite{ entropy_cohe}. By virtue of its growing importance and application, from interpreting pure states to algorithms, coherence is getting much more attention nowadays \cite{cohe_algo}.

Coherence is a prominent feature of quantumness while measurement can be used for the scaling of quantumness.
It is important to say that the foundation of quantum mechanics is built upon the measurement postulates.
Quantum systems are fragile to measurement, and so the understanding of the measurement process is always crucial in the quantum world. It should be reminded that when we intend to measure a quantum state, we lose its coherence completely. There is an engrossing way to keep the system of interest without completely losing the coherence, and so the quantumness, during measurement. The new scheme, which weakly couples the measurement device with the system, does not collapse the state vector, rather biased by a small angle such that the measurement device does not show a clear eigenvalue but a superposition of several \cite{weak_mmt, weak2}. Hence, we can say that a weak value of an observable can take values outside its spectrum. It shall be noted that weak measurement is universal in the sense that a generalized measurement can be perceived as a sequence of weak measurements. This ascendancy of weak measurements has been utilized in many areas.
There exists a close relation between uncertainty of measurements and coherence of a quantum state. Hence, uncertainty is being widely studied along with coherence \cite{Cohe_uncer, part_cohe, coh_Vs_Uncer}.

In this work, we investigate the use of weak measurements via Hellinger distance to calculate geometric coherence. We establish the relationship between the proposed coherence measure and the well-known geometric coherence based on von Neumann projective measurement. Moreover, we demonstrate that the proposed coherence measure satisfies all the axioms of a good measure of correlation in the weak measurement regime. Additionally, we explore the uncertainty of a quantum state relative to weak measurement by calculating weak variance.   

\section{Coherence relative to measurement}\label{sec2}
In order to define a coherence measure based on the Hellinger distance relative to weak measurement, we first recall the definition of a coherence measure based on the Hellinger distance. The Hellinger distance between states $\sigma_1$ and $\sigma_2$ is defined as
\begin{equation*}
\mathcal{D}_H(\sigma_1, \sigma_2) = \lVert \sqrt{\sigma_1}-  \sqrt{\sigma}_2\rVert^2
\end{equation*}
where $ \|\mathcal{O} \| = \sqrt{\text{tr}(\mathcal{O}^\dagger \mathcal{O})}$ represents the Hilbert-Schmidt norm of an operator $ \mathcal{O}$. One can define coherence of a state $\rho$ as
\begin{equation}
C_H(\rho| \sigma) =  \lVert \sqrt{\rho}-  \sqrt{\sigma}\rVert^2
\end{equation}
where $\sigma \in \mathcal{I}$ is a set of incoherent states.
We take a special class of coherence where incoherent states are produced by measurements.  
Consider a bipartite quantum state $ \rho $ shared between parties $a$ and $b$, partial coherence relative to measurement is given by \cite{coh_Vs_Uncer, part_cohe2, part_cohe3}
\begin{equation}
  C_H(\rho|M) := \| \sqrt{\rho} - M(\sqrt{\rho}) \|^2\label{h_cohe}
\end{equation}
where $M$ corresponds to the measurement over the state $\rho$. 
This definition of geometric coherence based on Hellinger distance can be related to MIN based on Hellinger distance (H-MIN) which is given by \cite{hmin}
\begin{equation}
 N_H(\rho) =\, ^{\text{max}}_{\Pi^a} \,C_H(\rho|\Pi). \label{h-min} 
\end{equation}
Here, $\Pi=\lbrace\Pi^a_k\rbrace = \lbrace\vert k \rangle \langle k \vert \rbrace$ corresponds to locally invariant von Neumann projective measurements of the subsystem $a$, and $ \Pi^a(\rho) = \sum_k(\Pi_k^a \otimes \mathds{1})\,\rho\,(\Pi_k^a \otimes \mathds{1})$. In the case of eq.(\ref{h_cohe}), measurement need not be locally invariant von Neumann projective measurements but the eq.(\ref{h-min})  is true only when the measurements are locally invariant.
 Then without loosing generality eq.(\ref{h_cohe}) can be rewritten as 
 \begin{equation}
 C_H(\rho| \Pi)=  1 - \text{tr}\big(\sqrt{\rho}\,\Pi^a(\sqrt{\rho})\big)
\end{equation}
where an equivalent representation is obtained in terms of skew information as 
\begin{equation}
  C_H(\rho|\Pi) =  \sum_i{I}(\rho,\Pi_i^a \otimes \mathds{1})
\end{equation}
where the Wigner-Yanase skew information $ {I}(\rho,K) = -\frac{1}{2}\text{tr}[\sqrt{\rho},K]^2$  \cite{skew_info}.
\section{Coherence relative to weak measurements}\label{sec3}
  A quantum measurement with any number of outcomes can be constructed as a sequence of two outcomes. Such weak measurements operators are defined as \cite{weak2}
\begin{equation}
  \Omega_x = \tau_1 \Pi^1 + \tau_2 \Pi^2, ~~\Omega_{-x} = \tau_2 \Pi^1 + \tau_1 \Pi^2
\end{equation}\label{weak_op}
where $\tau_1 = \sqrt{\frac{1-\tanh{x}}{2}},~ \tau_2 = \sqrt{\frac{1+\tanh{x}}{2}}$ with $x \in \mathbb{R}$ being the strength of weak measurements and $\Pi^1$ and $\Pi^2$ are the two orthogonal projectors such that $\Pi^1 + \Pi^2 = \mathds{1}$. Here $\Pi^1$ and $ \Pi^2$ can be decomposed as $ \Pi^1 = \sum^k_{i = 1} \Pi_i,\, \Pi^2 = \sum^n_{i = k+1} \Pi_i$, where $\{\Pi_i\}$ are von Neumann measurements \cite{weak_corltn}. Weak measurements obey the relation $\sum_{j = \pm x} \Omega_j \Omega_j^{\dagger} = \mathds{1}$. The post-measured state after a weak measurement on a bipartite state is given by \cite{hmin}
\begin{align}
  \Omega(\rho) &= \sum_{k = \pm x}(\Omega_k \otimes \mathds{1})\,\rho \,(\Omega_k \otimes \mathds{1})\\
               &= \tau\rho + (1-\tau)\Pi^a(\rho)\label{weak1}
\end{align}
where $\tau = 2 \tau_1\tau_2 = \text{sech} x$. In the asymptotic limit $x\rightarrow \infty$, weak measurements operators reduce to orthogonal projectors and has the effect of an identity operator in the limit $ {x \rightarrow 0}$. 
With this we define coherence measure relative to weak measurements for a bipartite state $ \rho$ as
\begin{align}
  C_H(\rho|\Omega) &= \| \sqrt{\rho} - \Omega(\sqrt{\rho}) \|^2. \label{weak_cohrce}                
\end{align}
Alternatively, the above equation can be represented as
\begin{equation}
C_H(\rho|\Omega) = \sum_{k = \pm x} I(\rho,\Omega_k \otimes \mathds{1})
\end{equation}
Maximizing $ C_H(\rho|\Omega)$ gives the MIN based on Hellinger distance quantified by weak measurements as \cite{hmin}
\begin{equation}
  N_w(\rho) = \,^{\text{max}}_{\Omega}C_H (\rho|\Omega).
\end{equation}
\begin{theorem}
  For any bipartite state $\rho$, the coherence relative to weak measurements is given as
  \begin{equation*}
    C_H(\rho|\Omega) = (1-\tau)^2 C_H(\rho|\Pi)
  \end{equation*}
  where $ \Pi $ corresponds to $  \Pi^a(\rho) = \sum_k(\Pi^a_k \otimes \mathds{1})\,\rho\,(\Pi^a_k \otimes \mathds{1})$. 
\end{theorem} 
\begin{proof}
From the definition of coherence relative to weak measurements,
\begin{align} \nonumber
C_H(\rho|\Omega) &= \|\sqrt{\rho} - \Omega(\sqrt{\rho})\|^2 \\ \nonumber
               &= \text{tr}[\sqrt{\rho}-\Omega(\sqrt{\rho})][\sqrt{\rho}-\Omega(\sqrt{\rho})]^\dagger \\ \nonumber
               &= \text{tr}[\rho(\tau^2 - 2\tau +1) + \sqrt{\rho}\Pi^a(\sqrt{\rho})(2\tau -1-\tau^2)]\\ \nonumber
               &= (1- \tau)^2 \text{tr}[\rho - \sqrt{\rho}\,\Pi^a(\sqrt{\rho})]\\ \label{weak_coh1}
               &= (1-\tau)^2 C(\rho|\Pi).
\end{align}   
\end{proof}

It shall be noted that the quantity $ (1-\tau)^2$ can be regarded as the strength ratio of coherence relative to measurement.
To prove the above measure as a faithful quantifiers of coherence we have to show that the measure should satisfy the axioms of a valid coherence measure.
\begin{itemize}
  \item[($\mathcal{C}1$)] \textit{Positivity}: $C_H(\rho|\Omega) \geq 0$, with equality holds for $ \rho \in \mathcal{I}_p$ under the condition that, the strength of measurement $ x > 0 $. Here $\mathcal{I}_p$ is the set of  partial incoherent states such that \cite{part_cohe}
  \begin{equation}
    \mathcal{I}_p = \{\rho : \Pi^a(\rho) = \rho\}.
  \end{equation}
  To be called as partially incoherent, a completely positive and trace preserving (CPTP) map $\Phi$ with Kraus operators $A_k$ satisfies $ A_k \mathcal{I}_p A^\dagger_k \in \mathcal{I}_p$. The proof is as follows:
  \begin{proof}
Considering the case if $ \rho \in \mathcal{I}_p$ such that $ \sqrt{\rho} = \Pi^a(\sqrt{\rho})$, then by eq.(\ref{weak1}) we have $ \Omega(\sqrt{\rho}) = \sqrt{\rho}$ hence $C_H(\rho|\Omega)$ vanishes for $ x \geq 0$. The converse is also valid if $C_H(\rho|\Omega) = 0$, then we have $ \sqrt{\rho} = \Omega(\sqrt{\rho})$, such that for $ x \geq 0$,
\begin{equation*}
 \tau \sqrt{\rho} + (1-\tau)\Pi^a(\sqrt{\rho}) = \sqrt{\rho} = \Pi^a(\sqrt{\rho}),
\end{equation*}
then we can say $ \rho \in \mathcal{I}_p$.
\end{proof}
\item[($\mathcal{C}2a$)] \textit{Weak monotonicity}: Monotonicity under partial incoherent completely positive trace preserving maps $C_H(\rho|\Omega) \geq C(\Phi[\rho]|\Omega)$ with $ \Phi[\mathcal{I}_p] \subseteq \mathcal{I}_p$.\\
\begin{proof}
Consider $ d(\rho,\sigma) = \text{tr}(\sqrt{\rho}\,\sqrt{\sigma})$. An incoherent CPTP map is defined as $ \Phi(\rho) = \sum_k A_k \rho A_k^\dagger$ with Kraus operators satisfying $\sum_k A_k^\dagger A_k = \mathds{1}$ and $ A_k \mathcal{I}_p A_k^\dagger \subseteq \mathcal{I}_p$. These maps can be expressed in terms of unitary operator $U$ as 
\begin{equation}
\Phi(\rho) = \text{tr}_E\left(U\left(\rho \otimes\lvert e \rangle \langle e \rvert\right) U^\dagger\right)
\end{equation}
 where the environment is written in the basis $\lvert e \rangle$. With this we have, 
\begin{align*}
  d[\Phi(\rho),\Phi(\sigma)] &= d[\text{tr}_E\left(U (\rho\otimes\lvert e \rangle \langle e \rvert) U^\dagger\right), \text{tr}_E\left(U (\sigma \otimes\lvert e \rangle \langle e \rvert) U^\dagger\right)]\\
  &\geq d\left( \rho \otimes \lvert e \rangle \langle e \rvert , \sigma \otimes \lvert e \rangle \langle e \rvert\right).
\end{align*}
 By the property of trace, we employ the relation $ d (\rho_k, \sigma_k) \geq d(\rho, \sigma)$, with $ \rho_k$ and $ \sigma_k$ are the marginal states of $\rho$ and $\sigma$ respectively.
\begin{equation*}
  d[\Phi(\rho),\Phi(\sigma)] \geq d(\rho,\sigma).
\end{equation*}
With this, we have
\begin{equation*}
 1- \text{tr}[\Phi(\rho),\Phi(\sigma)] \leq 1- \text{tr}(\rho,\sigma)
\end{equation*}
implying that
\begin{equation*}
 C_H(\rho|\Omega) \geq C \left(\Phi(\rho)|\Omega\right).
\end{equation*}
\end{proof}
\item[($\mathcal{C}2b$)] \textit{Strong monotonicity}: Monotonocity under selective partial incoherent operation on average, i.e $C_H(\rho|\Omega) \geq \sum_k p_k C_H(\rho_k|\Omega)$ where $ \rho_k = A_k \rho A^\dagger_k / p_k$ with $ p_k $ being $ \text{tr}( A_k \rho A^\dagger_k)$ for a partial incoherent operator $ A_k$ satisfying $ A_k \mathcal{I}_pA^\dagger_k \in \mathcal{I}_p$.
\begin{proof}
Consider $ \text{tr}(\sqrt{\rho}\,\Omega(\sqrt{\rho})) = d(\rho,\Omega)$ and let us recall the Stinespring dilation theorem \cite{stinespring1,stinespring2} which states that $ T: \mathcal{S}(\mathcal{H}) \rightarrow \mathcal{S}(\mathcal{H})$ be a CPTP map on finite dimensional Hilbert space $ \mathcal{H}$, then there exist another Hilbert space $\mathcal{K}$ and unitary operator $U$ on $ \mathcal{H} \otimes \mathcal{K}$ such that
\begin{equation}
  T(\rho) = \text{tr}_\mathcal{K}[U(\rho \otimes \lvert 0 \rangle \langle 0 \rvert)U^\dagger]
\end{equation}
$ \forall \rho \in \mathcal{S}(\mathcal{H})$. $ \mathcal{K}$ is chosen such that $\text{dim} \,\mathcal{K} \leq \text{dim} \,2\mathcal{H}$. The Stinespring representation can be used for the Kraus decomposition as \cite{stinespring1}
\begin{equation*}
 A_k\rho A^\dagger_k = \text{tr}_\mathcal{K}\left[\mathds{1}\otimes \lvert k \rangle \langle k \rvert U(\rho \otimes \lvert 0 \rangle \langle 0 \rvert)U^\dagger \mathds{1}\otimes \lvert k \rangle \langle k \rvert\right],
\end{equation*}
where $ \sum_k A_k \rho A^\dagger_k $ be a CPTP map and an ancillary state $ \lvert 0 \rangle \langle 0 \rvert \in \mathcal{K}$. Then we have \\
$
     \sum_k d \bigg(A_k\rho A^{\dagger}_k,A_k\sigma A^{\dagger}_k\bigg) =
      \sum_k d \bigg(\text{tr}_\mathcal{K}\big[ \mathds{1}\otimes\lvert k \rangle \langle k \rvert U(\rho \otimes\lvert 0 \rangle \langle 0 \rvert)U^{\dagger}\mathds{1}\otimes\lvert k \rangle \langle k \rvert \big],
$

$ \hspace{5cm}     
 \text{tr}_\mathcal{K}\big[\mathds{1}\otimes\lvert k \rangle \langle k \rvert U(\sigma \otimes\lvert 0 \rangle \langle 0 \rvert)U^{\dagger}\mathds{1}\otimes\lvert k \rangle \langle k \rvert \big]\bigg),
   $
\\where $\sigma = \Omega(\sqrt{\rho}). $ By using the monotonicity of $d(\rho,\sigma)$ and by the property $ d(\sum_i \Pi \rho \Pi_i \sum_i \Pi_i \sigma \Pi_i ) = \sum_i d(\Pi_i \rho \Pi, \Pi_i \sigma \Pi)$ for $ \{\Pi_i\}$ being set of orthogonal projectors, we have \\
$
  \sum_k d(A_k \rho A^\dagger_k, A_k \sigma A^\dagger_k) \geq \sum_k d \bigg( \mathds{1}\otimes\lvert k \rangle \langle k \rvert U(\rho \otimes\lvert 0 \rangle \langle 0 \rvert)U^{\dagger}\mathds{1}\otimes\lvert k \rangle \langle k \rvert , 
$

$ \hspace{5cm}  
   \mathds{1}\otimes\lvert k \rangle \langle k \rvert U(\sigma \otimes\lvert 0 \rangle \langle 0 \rvert)U^{\dagger}\mathds{1}\otimes\lvert k \rangle \langle k \rvert \bigg)
   $
   
   $ \hspace{4cm}
 \geq d(U \rho \otimes \rvert 0\rangle \langle 0\rvert U^{\dagger},U \sigma \otimes\rvert 0\rangle \langle 0\rvert U^{\dagger})$.
  \\We impose the unitary invariance and trace property, and thus 
  \begin{equation*}
   d(A_k \rho A^{\dagger}_k, A_k \Omega A^{\dagger}_k) \geq d (\rho,\Omega),
  \end{equation*}
with which we have 
   \begin{equation*}
     \sum_k p_k C_H(\rho_k|\Omega) \leq C_H(\rho|\Omega).
   \end{equation*}
   \end{proof}
\item[($\mathcal{C}3$)] \textit{Convexity}: The coherence measure is non-increasing under mixing of quantum states such that $ C_H(\sum_ip_i\rho_i|\Omega) \leq \sum_i p_i C_H(\rho_i|\Omega)$ for any set of sates $\{\rho_i\}$ with $ \sum_i p_i = 1$.
\begin{proof}
The quantity $ d (\rho,\Omega) = \text{tr}\left(\sqrt{\rho}\,\Omega(\sqrt{\rho})\right)$ can be considered as the affinity which is concave in nature such that \cite{part_cohe},
\begin{equation*}
  d\left(\sum_k p_k \rho_k|\Omega\right) \geq \sum_k p_k d(\rho_k|\Omega).
\end{equation*}
Since $C_H(\rho|\Omega) = 1- d(\rho,\Omega)$ implying
\begin{equation*}
   C_H\left(\sum_i p_i\rho_i|\Omega\right) \leq \sum_i p_i C_H(\rho_i|\Omega).
\end{equation*}
\end{proof}
\item[($\mathcal{C}4$)] In the $ \lim_{x \rightarrow 0} C_H(\rho|\Omega) = 0$ and $ \lim_{x \rightarrow \infty} C_H(\rho|\Omega) = C_H(\rho|\Pi)$ we have
\begin{equation*}
0 \leq C_H(\rho|\Omega) \leq C_H(\rho|\Pi).
\end{equation*}
\end{itemize} 
In light of the properties $ \mathcal{C}1$ to $ \mathcal{C}3$, we claim that the proposed quantity is a faithful measure of quantum coherence relative to weak measurements.
\begin{theorem}
  For a non-zero measurement strength $ x >0$, $C_H(\rho|\Omega) = C_H(\rho|\Pi) = 0$ iff $ \rho$ is a product state.
\end{theorem}
\begin{proof}
For $\sqrt{\rho} = \rho^a \otimes \rho^b$ being a product state, there exist a spectral decomposition of the form $\sqrt{\rho} = \sum_i p_i \Pi_i^a \otimes \Pi_i^b$.
 From the definition we have $ C_H(\rho|\Omega) = 1 - \text{tr}[\sqrt{\rho}\,\Omega(\sqrt{\rho})]$. Further
\begin{align*}
  \Omega(\sqrt{\rho}) &= \tau  \sum_i p_i \Pi_i^a \otimes \Pi_i^b+ (1-\tau)\Pi^a\left( \sum_i p_i \Pi_i^a \otimes \Pi_i^b\right) \\                    
                      &= \sqrt{\rho}
\end{align*}
implying $ C_H(\rho|\Omega) = C_H(\rho|\Pi) = 1 - \text{tr}(\rho) = 0$. 
The converse is also true assuming that
\begin{equation*}
 C_H(\rho|\Omega) = C_H(\rho|\Pi) = 0.
\end{equation*}
Then, $ \text{tr}(\sqrt{\rho}\,\Omega(\sqrt{\rho})) =\text{tr}(\sqrt{\rho}\,\Pi^a(\sqrt{\rho}))  = 1$ implying
$ \Omega(\sqrt{\rho}) = \Pi^a(\sqrt{\rho}) = \sqrt{\rho}$. This happens iff
\begin{equation*}
\sqrt{\rho} = \sum_i p_i \Pi_i^a \otimes \Pi_i^b = \rho^a \otimes \rho^b.
\end{equation*}
\end{proof} 
The coherence relative to weak measurements, as defined in eq.(\ref{weak_coh1}), is a monotonically increasing function with respect to the strength of weak measurement, denoted by $x$. Specifically, we have $C_H(\rho|\Omega) \leq C_H(\rho|\Pi)$, and equality holds as $x \rightarrow \infty$.

In the following analysis, we explore the calculation of coherence relative to weak measurements for well-known families of quantum states, including Bell diagonal states and Werner states.
\subsection*{Bell diagonal state}
The Bloch vector representation of Bell diagonal state is given as 
\begin{equation}
  \rho^{BD} = \frac{1}{4}\Big(\mathbb{I}\otimes \mathbb{I} + \sum^{3}_{i = 1} c_i(\sigma_i \otimes \sigma_i)\Big) \label{bell_diag}
\end{equation}
where $\textbf{c} = (c_1, c_2, c_3) $ is the correlation vector with coefficients with $ -1 \leq c_i \leq 1$. Then
\begin{equation}
  \sqrt{\rho^{BD}} =  \frac{1}{4}\Big(\delta~\mathbb{I}\otimes \mathbb{I} + \sum^{3}_{i = 1} d_i(\sigma_i \otimes \sigma_i)\Big)
\end{equation}
where $\delta = \text{tr}(\sqrt{\rho^{BD}}) = \sum_i \sqrt{\lambda_i}$ and
\begin{equation*}
  d_1 = \sqrt{\lambda_1}- \sqrt{\lambda_2} + \sqrt{\lambda_3} - \sqrt{\lambda_4} 
\end{equation*}
\begin{equation*}
d_2 =  -\sqrt{\lambda_1}+ \sqrt{\lambda_2} + \sqrt{\lambda_3} - \sqrt{\lambda_4} 
\end{equation*}
\begin{equation*}
d_3 = \sqrt{\lambda_1} + \sqrt{\lambda_2} - \sqrt{\lambda_3} - \sqrt{\lambda_4}. 
\end{equation*}
Here $\lambda_i$ are the eigenvalues of the Bell diagonal state. 
\subsection*{Werner state}
Werner state with $d\times d$ dimension can be represented as \cite{werner}
\begin{equation}
  \rho^{w} = \frac{d-y}{d^3-d}\mathbb{I} + \frac{yd-1}{d^3-d}\sum_{\alpha \beta}\lvert \alpha \rangle \langle\beta \rvert  \otimes\lvert \beta \rangle \langle \alpha  \rvert
\end{equation}
where $\sum_{\alpha \beta}\lvert \alpha \rangle \langle\beta \rvert  \otimes\lvert \beta \rangle \langle \alpha  \rvert$ is flip operator with $y \in [-1,1]$.
\begin{figure}[!ht]
 \includegraphics[width=0.45\linewidth]{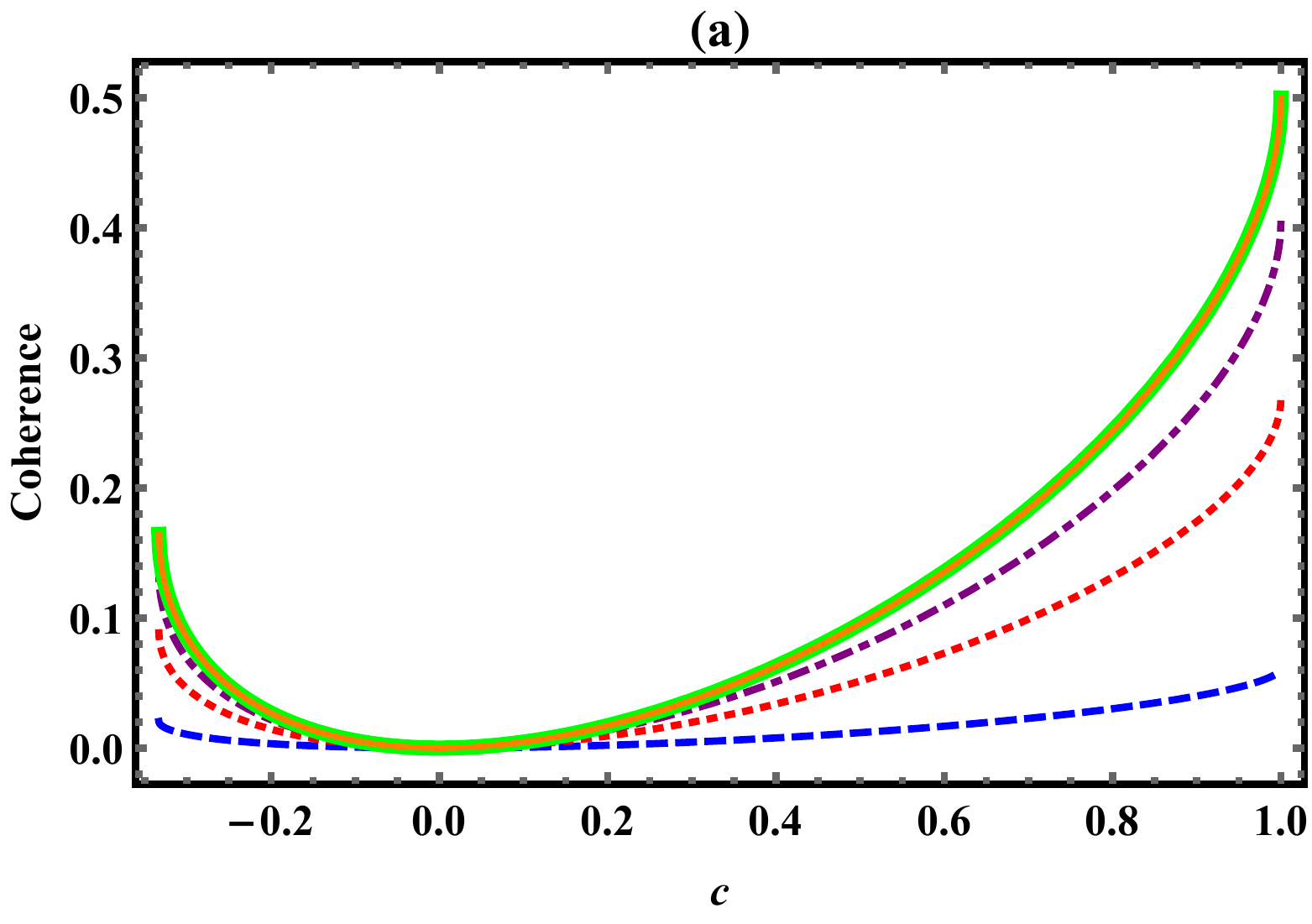}
 \includegraphics[width=0.45\linewidth]{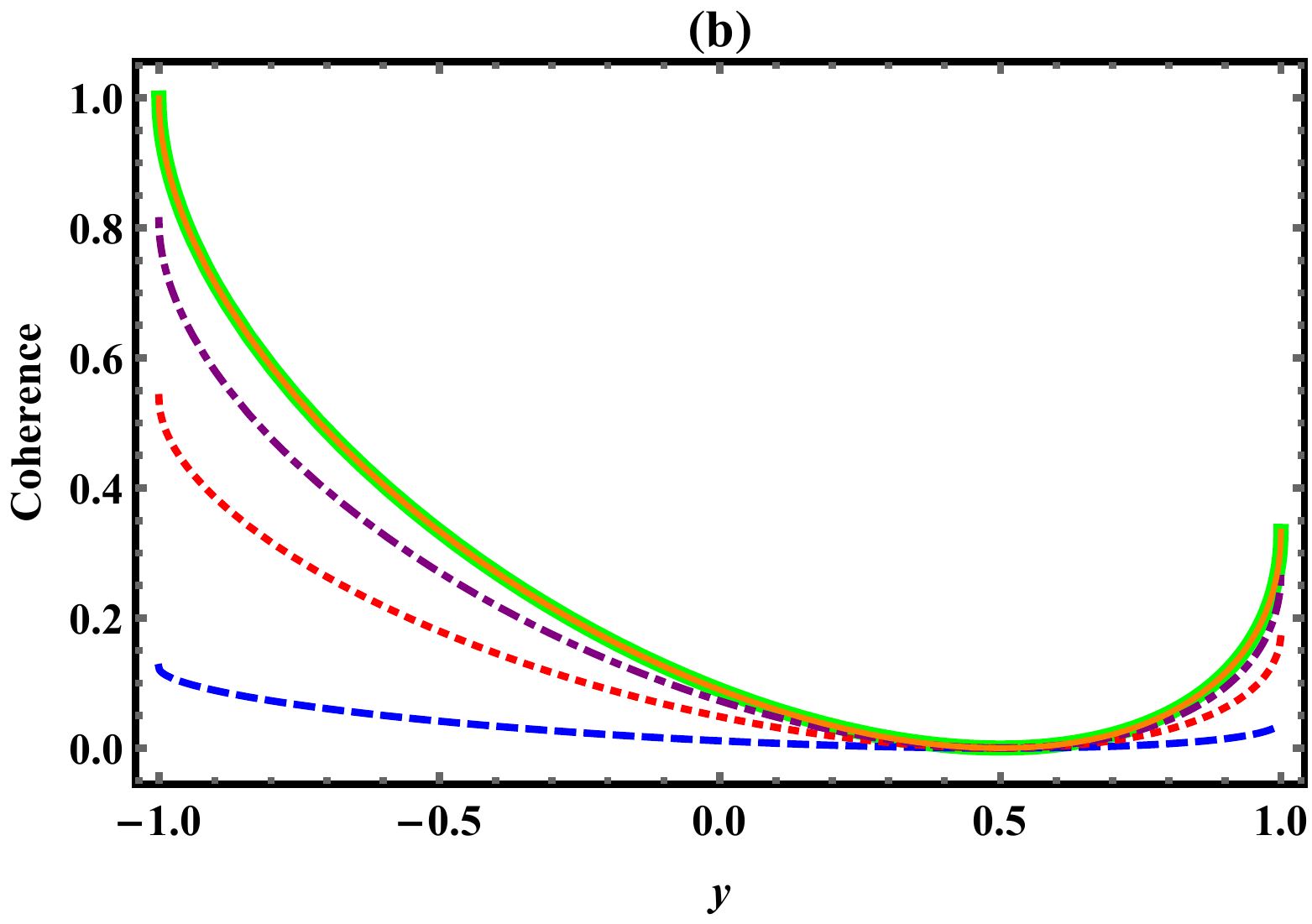}\\
\caption{Coherence quantified by weak measurement for strength $ x=1$ (dashed), $ x = 2$ (Dotted), $ x = 3$ (Dot-dashed), $ x = 50$ (Sold-thin), von Neumann measurement (Solid-thick),   for (a) Bell states with $ c_1 = c_2 = c_3 = -c$ (b) Werner states.}
\label{cohe_examp} 
\end{figure}

In Fig.[\ref{cohe_examp}], we analyze the impact of measurement strengths, denoted by $x$, on coherence for Bell and Werner states. Notably, we observe that as the measurement strength increases, weak coherence approaches coherence relative to von Neumann measurement. Intriguingly, the coherence relative to von Neumann measurement converges to unity as the measurement strength becomes large ($x = 50$). This outcome is consistent with the understanding that weak measurement transitions to von Neumann measurement in the asymptotic limit ($x \rightarrow \infty$).
\section{Quantum correlation via coherence}\label{sec5}
Consider a bipartite state shared between parties $a$ and $b$, and we define a correlation function
\begin{equation}
  \Delta_w(\rho|\Omega) = C_H(\rho|\Omega) - C_H(\rho^a\otimes\rho^b|\Omega) \label{delta}
\end{equation}
as the difference between global state coherence and its coherence of marginal product state. The coherence relative to measurements is represented in terms of skew information and 
\begin{align*}
  C_H(\rho^a\otimes\rho^b|\Omega) &=  \sum_k{I}(\rho^a \otimes \rho^b,\Omega_k \otimes \mathds{1}) \\
                                &=  \sum_k{I}(\rho^a ,\Omega_k )\\
                                &= C_H(\rho^a|\Omega).
\end{align*}
Hence eq.(\ref{delta}) can be modified as $ \Delta_w(\rho|\Omega) = C_H(\rho|\Omega) - C_H(\rho^a|\Omega)$. Defining quantum correlation induced by coherence measure as
\begin{equation}
 \mathcal{Q}_w(\rho) =\, ^\text{min}_\Omega \Delta_w(\rho|\Omega)
\end{equation}
 where $ \mathcal{Q}_w$ refers to quantum correlation due to weak measurements.
This quantity possesses the following properties and satisfy the axioms of a faithful quantifier of quantum correlation.
 \begin{itemize}
 \item[(i)] $ \mathcal{Q}_w(\rho) \geq 0$, with equality holds for a classically correlated state i.e., $ \rho = \Pi^a(\rho)$ with the measurement strength $ x > 0$.
\begin{proof}
By the definition
\begin{align*}
\mathcal{Q}_w(\rho) &=\,^{\text{min}}_{\Omega} \Delta_w(\rho|\Omega)\\
                    &=\,^{\text{min}}_{\Omega}\left( C_H(\rho|\Omega) - C_H(\rho^a|\Omega)\right).
\end{align*}
This can be written as
\begin{equation*}
 \left(1- \text{tr}[\sqrt{\rho}\,\Omega(\sqrt{\rho})]\right) - \left(1-\text{tr}[\sqrt{\rho^a}\,\Omega(\sqrt{\rho})] \right) \geq 0.
\end{equation*}
Since $ \text{tr}[\sqrt{\rho^a}\,\Omega(\sqrt{\rho^a})] \geq \text{tr}[\sqrt{\rho}\Omega(\sqrt{\rho})]$ \cite{skew_prop}, we have the proof.
\end{proof}
\item [(ii)] $\mathcal{Q}_w(\rho)$ is non-increasing under CPTP map $ \Phi$ as $ \mathcal{Q}_w\left([\mathds{1}\otimes \Phi]\rho\right) \leq \mathcal{Q}_w(\rho)$
\begin{proof}
For any CPTP map $ \Phi$ acting on the subsystem $b$, we are able to choose an ancillary state $\rho^c$ and an unitary operator $U$ such that $ \Phi(\rho^b) = \text{tr}_c \left(U(\rho^b\otimes\rho^c)U^\dagger\right)$. Here, 
\begin{align*}
  \Delta_w\left((\mathds{1}\otimes \Phi)\rho|\Omega\right) &= C_H\left((\mathds{1}\otimes\Phi)\rho|\Omega\right)-C_H(\rho^a|\Omega) \\
  &= \sum_k {I}\left((\mathds{1}\otimes \Phi)\rho,\Pi_k^a \otimes \mathds{1}\right) - {I}(\rho^a,\Pi_k^a).
\end{align*}
The above form can be rewritten using the ancilla as
\begin{align*}
  {I}\big((\mathds{1}\otimes \Phi)\rho,\Omega\otimes\mathds{1}\big) =& {I}\big(\text{tr}[(\mathds{1}\otimes U)(\rho\otimes\rho^c) (\mathds{1}\otimes U)^\dagger], \Omega \otimes \mathds{1} \big)\\
  &\leq {I}\big((\mathds{1}\otimes U)(\rho\otimes\rho^c)(\mathds{1}\otimes U)^\dagger,\Omega \otimes\mathds{1}\otimes\mathds{1}^c\big)\\
  &= {I}\big(\rho\otimes\rho^c,\Omega \otimes\mathds{1}\otimes\mathds{1}^c \big)\\
  &= {I}\big(\rho,\Omega \otimes\mathds{1} \big).
\end{align*}
From this the coherence based on skew information can be written as 
\begin{align*}
  \Delta_w\big((\mathds{1}\otimes\Phi)\rho|\Omega\big) \leq \Delta_w(\rho|\Omega).
\end{align*}
From this it is clear that $ C\big((\mathds{1}\otimes\Phi)\rho|\Omega\big) \leq C_H(\rho|\Omega)$ and hence
\begin{equation*}
  \mathcal{Q}_w\left([\mathds{1}\otimes \Phi]\rho\right) \leq \mathcal{Q}_w(\rho).
\end{equation*}
\end{proof}
\item[(iii)] $ \mathcal{Q}_w(\rho)$ is invariant under local unitary operation such that  $ \mathcal{Q}_w(U\rho U^\dagger) = \mathcal{Q}_w(\rho)$ where $ U = U^a \otimes U^b$ with $U^{a,b}$ being the unitary operation on the marginal state $a,\,(b)$. 
\begin{proof}
Consider $ d(\rho,\Omega) = \text{tr}\big(\sqrt{\rho}\,\Omega(\sqrt{\rho})\big)$ such that the unitary Inverclyde of $d(\rho,\Omega)$ follows,
\begin{equation}
\text{tr}[U_{{\rho}} \Omega(U_\rho)] = \text{tr}[U_{\rho} \big(\tau U_\rho +(1-\tau)\{\Pi^a\otimes\mathds{1} (U_\rho)\Pi^a\otimes\mathds{1}\}\big)]
\end{equation}
where $ U_\rho =U\sqrt{\rho}U^\dagger $.


Making use of trace property we have,
\begin{align}\nonumber
 \text{tr}[\tau\rho + (1-\tau)\left(U_{\Pi^a}\sqrt{\rho} U_{\Pi^a } \sqrt{\rho}\right)] &= \text{tr}[\tau\rho + (1-\tau)\big(\Pi^a(\sqrt{\rho})\sqrt{\rho}\big)] \\
 &=\text{tr}[\sqrt{\rho}\,\Omega(\sqrt{\rho})]
\end{align}
where $U_{\Pi^a}= U^\dagger \Pi^a U \otimes \mathds{1}$.
It follows that, since the quantity $\text{tr}[\sqrt{\rho}\,\Omega(\sqrt{\rho})] $ is unitary invariant, $ \mathcal{Q}_w(\rho)$ is invariant under unitary transformation.
\end{proof}
\end{itemize}
\section{Quantum uncertainty under weak measurement}\label{sec6}
While measurement serves as a crucial factor in determining physical reality, uncertainty arises from the inherent limitations of obtaining precise information in the quantum realm. Specifically, the total uncertainty of an observable $\mathcal{O}$ in a given state $\rho$ can be quantified using the framework proposed in \cite{total_uncertainty, coh_Vs_Uncer}.
\begin{equation}
  \mathcal{V}(\rho,\mathcal{O}) = C(\rho,\mathcal{O}) + Q(\rho,\mathcal{O})\label{tot_uncer}
\end{equation} 
where $ \mathcal{V}(\rho,\mathcal{O})$ denotes the variance of an observable with respect to the state which consists of classical ($ C(\rho,\mathcal{O})$) and quantum mechanical $(Q(\rho,\mathcal{O})) $ counterparts of uncertainties.
The above equation is also valid for an operator $\mathcal{A}$ \cite{operator_varince} and for a Hermitian operator $\mathcal{A}$, the total uncertainty of the operator in the state $\rho$ is given by
\begin{equation}
\mathcal{V}(\rho,\mathcal{A}) = \text{tr}(\rho\,\mathcal{A}^2_0) \label{varience}
\end{equation}
where $ \mathcal{A}_0 = \mathcal{A}-\text{tr}(\rho\,\mathcal{A})$.
If $ \rho$ is a pure state, the contribution of classical uncertainty vanishes. For any quantum measurement represented by a POVM $ M = \{ M_k; k = 0, 1, \cdots, m\}$, the quantum uncertainty of $M$ with respect to the state $\rho$ is given by \cite{coh_Vs_Uncer}
\begin{equation}
  Q(\rho,M) = \sum_i {I}(\rho,M_i)
\end{equation}
where the measurement $M$ plays an active role. The quantity mentioned above can be considered as bonafide measure of coherence when the state $\rho$ plays an active role.

 For a bipartite state $\rho$, the quantum uncertainty of the weak measurement operator is given by
\begin{equation}
  Q(\rho,\Omega \otimes \mathds{1}) = \sum_i{I}(\rho,\Omega_i \otimes \mathds{1}) \label{qua_uncer}.
\end{equation}
Here the weak measurement has an active role to play. There exists a dual
interpretation of the quantity $Q(\rho, M)$ such that, it is considered as coherence in the quantum state point of view and uncertainty when measurement is considered \cite{coh_Vs_Uncer}. The defined uncertainty depends on the strength of the measurement such that, one could define the least quantum uncertainty possible. From the definition of skew information, $ Q(\rho, \Omega \otimes \mathds{1}) \geq 0$ with equality holds for $ [\rho, \Omega \otimes \mathds{1}] = 0$. Also, it holds 
\begin{equation*}
  Q(\rho,\Omega \otimes \mathds{1}) \geq Q(\rho^a, \Omega)
\end{equation*}
which shows that coherence decrease under partial trace operation \cite{skew_prop}. The classical part of total uncertainty of weak measurement in the state $\rho$ is given by 
\begin{equation}
{C}(\rho, \Omega) =\mathcal{V}(\rho, \Omega)- I(\rho,\Omega).
\end{equation}
Then from eq.(\ref{varience}), we have
\begin{align}
{C}(\rho, \Omega) &= \text{tr}(\rho\,\Omega^2)- \text{tr}(\rho\,\Omega)^2-\text{tr}(\Omega \rho \Omega) +\text{tr}(\sqrt{\rho}\Omega \sqrt{\rho} \Omega) \nonumber\\
                  &= \text{tr}(\sqrt{\rho}\,\Omega \sqrt{\rho} \,\Omega)  -\text{tr}(\rho\,\Omega)^2\nonumber\\
                  &= \text{tr}(\sqrt{\rho}\,\Omega \sqrt{\rho} \,\Omega) - \text{tr}(\Omega(\rho^2)).
\end{align}
\begin{theorem}
For any two quantum states $\rho$ and $\sigma$, we have
\begin{equation}
I(\rho,\Omega)I(\sigma, \Omega) \geq 4 \frac{|\text{tr}[\rho, \sigma]\Omega|^4}{(H^2 - 4 H )^2}
\end{equation}
where $H = \lVert \sqrt{\rho} - \sqrt{\sigma}\rVert^2$, the Hellinger distance between the states $\rho$ and $\sigma$.
\end{theorem}
\begin{proof}
Let $\alpha = \frac{i[\sqrt{\rho}, \Omega]}{\sqrt{2 I(\rho,\Omega)}}$, such that $ \text{tr} (\alpha^2) = 1$. With this,
\begin{eqnarray} \nonumber 
i [\sqrt{\rho}, \Omega] &= \sqrt{2I(\rho, \Omega)}\,\alpha\\ \nonumber.
\end{eqnarray}
Now, we have 
\begin{equation}
i \sqrt{\sigma} [\sqrt{\rho}, \Omega] = \sqrt{2I(\rho, \Omega)}\,\sqrt{\sigma}\alpha\\ \nonumber.
\end{equation}
Taking trace, we have
\begin{equation}
i\,\text{tr}([\sqrt{\rho}, \sqrt{\sigma}\,]\,\Omega) =  \sqrt{2I(\rho, \Omega)}\,\text{tr}(\sqrt{\sigma}\,\alpha).\label{base}
\end{equation}
By the Parseval inequality, 
\begin{eqnarray*}
\text{tr}(\sigma)  \geq &|\text{tr}(\sqrt{\rho}\sqrt{\sigma})|^2 + |\text{tr}(\sqrt{\rho}\sqrt{\alpha^2})|^2\\ 
1 \geq & |\text{tr}(\sqrt{\rho}\sqrt{\sigma})|^2 + |\text{tr}(\sqrt{\rho}\sqrt{\alpha^2})|^2.
\end{eqnarray*}
We now have the inequality
\begin{equation}
\sqrt{1- |\text{tr}(\sqrt{\rho}\sqrt{\sigma})|^2} \geq  |\text{tr}(\sqrt{\rho}\sqrt{\alpha^2})|.\label{parceval}
\end{equation}
From the definition of Hellinger distance, we have 
\begin{align*}
H(\rho, \sigma) &= \| \sqrt{\rho} - \sqrt{\sigma} \|^2\\
                &= \text{tr}(\rho) + \text{tr}(\sigma) - \text{tr}(\sqrt{\rho}\sqrt{\sigma} ) - \text{tr}(\sqrt{\sigma}\sqrt{\rho})\\
                &= 2\Big(1- \text{tr}(\sqrt{\rho}\sqrt{\sigma})\Big),
\end{align*}
then 
\begin{equation*}
  \text{tr}(\sqrt{\rho}\sqrt{\sigma}) = 1- \frac{H(\rho,\sigma)}{2}.
\end{equation*}
Using eq.(\ref{base}) and eq.(\ref{parceval}), we have 
\begin{equation}
I(\rho, \Omega) \geq \frac{{2}| \text{tr}[\sqrt{\rho},\sqrt{\sigma}]\Omega|^2}{H^2 - 4 H}.
\end{equation}
Here $ I(\rho, \Omega)$ represents the uncertainty of the weak measurements in the state $\rho$.
Similarly 
\begin{equation}
I(\sigma, \Omega) \geq \frac{{2}| \text{tr}[\sqrt{\rho},\sqrt{\sigma}]\Omega|^2}{H^2 - 4 H}.
\end{equation}
Hence the uncertainty relation can be written as
\begin{equation}
I(\rho,\Omega)I(\sigma, \Omega) \geq 4 \frac{|\text{tr}[\rho, \sigma]\Omega|^4}{(H^2 - 4 H )^2}.
\end{equation}
\end{proof}
Consider a quantum system which is pre-selected in the state $ \lvert \psi \rangle $ and weak measurement is done on an observable $ \mathcal{O}$. After the weak measurement, the system is post-selected in the state $ \lvert \phi\rangle$ which yields the weak value as
\begin{equation}
  \langle \mathcal{O} \rangle_w = \frac{\langle \phi\lvert \mathcal{O}\rvert \psi \rangle}{\langle \phi \lvert \psi \rangle }.
\end{equation}
Defining an operator $ \mathcal{W} $ as
\begin{equation*}
 \mathcal{W} = \frac{ \mathsf{P} \mathcal{O}}{\alpha},
\end{equation*}
where $ \mathsf{P} = \lvert \phi \rangle \langle \phi \rvert $ and $ \alpha = \lvert \langle \phi \rvert \psi \rangle \lvert^2$. It should be noted that $ \mathcal{W} \lvert \psi \rangle = \frac{ \langle \mathcal{O} \rangle_w \lvert \phi \rangle } {\langle \psi \lvert \phi \rangle} $ such that
\begin{equation}
  \langle \mathcal{W} \rangle = \langle \psi \lvert \mathcal{W} \rvert \psi \rangle = \text{tr}(\rho\mathcal{W}) = \langle \mathcal{O} \rangle_w.
\end{equation}
Variance of an operator $ \mathcal{W}$ is defined as 
\begin{equation}
\mathcal{V}_w(\rho,\mathcal{O}) = \langle \psi \lvert (\mathcal{W} - \langle \mathcal{W} \rangle)(\mathcal{W^\dagger} - \langle \mathcal{W^\dagger} \rangle) \rvert \psi \rangle \label{vari}
\end{equation}
with $   \langle \psi \lvert \mathcal{W} \rvert \psi \rangle =\langle \mathcal{W} \rangle$ and $   \langle \psi \lvert \mathcal{W^\dagger} \rvert \psi \rangle =\langle \mathcal{W} \rangle^*$. With this eq.(\ref{vari}) can be rewritten as 
\begin{equation}
  \mathcal{V}_w(\rho,\mathcal{O}) =\langle \psi \lvert \mathcal{W}  \mathcal{W^\dagger} \rvert \psi \rangle - \langle \psi \lvert \mathcal{W}  \rvert \psi \rangle \langle \psi \lvert \mathcal{W^\dagger} \rvert \psi \rangle.
\end{equation} 
If $ \mathcal{W}$ is Hermitian, the above equation becomes 
\begin{align}\nonumber
 \mathcal{V}_w(\rho,\mathcal{O}) &= \langle \mathcal{W}^2 \rangle -  \langle \mathcal{W} \rangle^2.
\end{align}
We have 
\begin{align*}
\langle \mathcal{W}^2 \rangle &= \text{tr}\left(\frac{\mathsf{P^2}\mathcal{O}^2}{\alpha^2}\right)\\
                    &= \frac{\langle \psi \lvert \phi \rangle \langle \phi \vert \mathcal{O}^2\rvert \psi \rangle}{\lvert \langle \phi \rvert \psi \rangle\rvert^4}\\
                    &= \frac{\langle \mathcal{O}^2 \rangle_w }{\alpha}.
\end{align*}
With this, the weak variance of an operator $\mathcal{O}$ is written as 
\begin{equation}
\mathcal{V}_w(\rho,\mathcal{O}) = \frac{\langle \mathcal{O}^2 \rangle_w}{\alpha}-\langle \mathcal{O} \rangle_w^2.\label{weak_uncert}
\end{equation}
The variance defined in the equation above in terms of weak values is known as the weak variance of the operator $\mathcal{O}$, which may not be a real quantity. This weak variance can be useful in quantifying the extent of partial collapse of quantum states under weak measurement.
\section{Summary and conclusion}
A coherence measure based on the weak measurement is proposed in this article. This measure generalizes the geometric coherence, which is quantified using Hellinger distance. We also identified weak H-MIN as a maximal weak coherence.
Least possible coherence produced as a result of infinitesimal disturbance on a quantum system is calculated using sequential weak measurement.
A correlation measure is defined with this coherence, which quantifies the correlation of weakly measured quantum system. Finally, we have introduced quantum uncertainty in the weak measurement regime, and the respective variance is computed in terms of weak values.


\begin{thebibliography}{99}
\bibitem{coh}T. Baumgratz,  M. Cramer and M. B. Plenio, \emph{Phys. Rev. Lett.} \textbf{113} (2014) 140401.
\bibitem{qkd}A. K. Ekert, \emph{Phys. Rev. Lett.} \textbf{67} (1991) 661.
\bibitem{assymetry1} G. Gour and R. W. Spekkens, \emph{New J. Phys.} \textbf{10} (2008) 033023.
\bibitem{quant_thrmo_1} M. Lostaglio, K. Korzekwa, D. Jennings, and T. Rudolph, \emph{Phys. Rev. X} \textbf{5} (2015) 021001.
\bibitem{quant_thrmo_2}M. Lostaglio, D. Jennings, and T. Rudolph, \emph{Nat. Commun.} \textbf{6} (2015) 6383.
\bibitem{quant_thrmo_3} P. Cwikli´nski, M. Studzi´nski, M. Horodecki, and J. Oppenheim, \emph{Phys. Rev. Lett.} \textbf{115} (2015) 210403.
\bibitem{quat_comm_1}H. L. Shi, S. Y. Liu, X. H. Wang, \textit{et al.}, \emph{Phys. Rev. A} \textbf{95}
(2017) 032307.
\bibitem{quat_comm_2}Y. L. Su, S. Y. Liu, X. H. Wang, \textit{et al.} \emph{Sci. Rep.} \textbf{8} (2018) 11081.
\bibitem{quat_comm_3}A. E. Rastegin, \emph{Quantum Inf. Process} \textbf{17} (2018) 179.
\bibitem{coh_met1}V. Giovannetti, S. Lloyd, L. Maccone, \emph{Science} \textbf{306} (2004) 1330.
\bibitem{coh_met2}R. D-Dobrzanski and L. Maccone,  \emph{Phys. Rev. Lett.} \textbf{113} 250801 (2014).
\bibitem{quant_bio1}M. Sarovar, A. Ishizaki, G. R. Fleming and K. B. Whaley \emph{Nat. Phys.} \textbf{6}, 462-467 (2010).
\bibitem{quant_bio2}S. F. Huelga and M. B. Plenio, \emph{Contemp. Phys.} \textbf{54} 181207 (2013).
\bibitem{quant_bio3}S. Lloyd, \emph{J. Phys. Conf. Ser.} \textbf{302} 012037 (2011).
\bibitem{quant_bio4}N. Lambert, Y.-N. Chen, Y.-C. Cheng, C.-M. Li, G.-Y. Chen and F. Nori, \emph{Nat. Phys.} \textbf{9} 10-18 (2013).
\bibitem{cohe_hellin}Z. X. Jin and S. M. Fei, \emph{Phys. Rev. A} \textbf{97} 062342.
\bibitem{muthu_cohe}R. Muthuganesan, V. K. Chandrasekar and R. Sankaranarayanan, \emph{Phys Lett. A} \textbf{394} (2021) 127205.
\bibitem{coh_form} A. Winter and D. Yang, \emph{Phys. Rev. Lett.} \textbf{116} (2016) 120404.
\bibitem{coh_entag}A. Streltsov, U. Singh, H. S. Dhar, M. N. Bera, and G. Adesso, \emph{Phys. Rev. Lett.} \textbf{115} (2015) 020403. 
\bibitem{rel_Ent_coh} K. Bu, U. Singh, S-M. Fei, A. K. Pati, J. Wu, \emph{Phys. Rev. Lett.} \textbf{119} 150405 (2017).
\bibitem{Tsalis_Coh}Zhao, H., Yu, C-S. \emph{Sci. Rep.} 8 299 (2018).
\bibitem{reneyi_cohe}H. Zhu, M. Hayashi, L. Chen, \emph{J. Phys. A: Math. and Theor.} \textbf{50} 47 (2017.
\bibitem{entropy_cohe}M. L. Guo, Z. X. J, B. L., B. Hu and S. M. Fei, \emph{Quantum Inf. Process} \textbf{19} (2020) 382.
\bibitem{cohe_algo}H. L Shi, Si-Y. Liu, X. H Wang, W. L Yang, Z. Y. Yang, and H. Fan, \emph{Phys. Rev. A} \textbf{95} (2017) 032307 .
\bibitem{weak_mmt}Y. Aharonov, D.Z. Albert, L. Vaidman, \emph{Phys. Rev. Lett.} \textbf{60} (1988) 14.
\bibitem{weak2}O. Oreshkov and T. A. Brun, \emph{Phys. Rev. Lett.} \textbf{95} (2005) 110409.
\bibitem{Cohe_uncer} S. Luo and Y. Sun, \emph{Commun. Theor. Phys.} \textbf{71} (2019) 1443–1447.
\bibitem{part_cohe}C. Xiong, A. Kumar, M. Huang, S. Das, U. Sen,5,and J. Wu1, \emph{Phys. Rev. A.} \textbf{99} (2019) 032305. 
\bibitem{coh_Vs_Uncer}S. Luo and Y. Sun, \emph{Phys. Rev. A.}  \textbf{96} (2017) 022130.
\bibitem{part_cohe2}Y. Sun, Y. Mao and S. Luo, \emph{Eur. Phys. Lett.} \textbf{118} (2017) 6007.
\bibitem{part_cohe3}S. Kim, L. Li, A. Kumar and J. Wu, \emph{Phys. Rev. A} \textbf{98} (2018) 023306.
\bibitem{hmin}Indrajith. V.S, R. Muthuganesan and R.Sankaranarayananan, \emph{Phyisca A} \textbf{566} (2021) 125615.
\bibitem{skew_info}E. P. Wigner and M. M. Yanase, \emph{Proc. Nat. Acad. Sci. U,S,A.} \textbf{49} (1963) 910-918.
\bibitem{weak_corltn}Y. Wang, J. Hou, and X. Qi, \emph{Int. J. Quant. Inf.} \textbf{15} (2017) 1750041.
\bibitem{stinespring1}W. F. Stinespring, \emph{Proc. Amer. Math. Soc.} \textbf{6} (1955) 211.
\bibitem{stinespring2}M. Keyl, \emph{Phys. Rep.} \textbf{369} (2002) 431.
\bibitem{werner}R. F. Werner, \emph{Phys. Rev. A.} \textbf{40} (1989) 4277. 
\bibitem{skew-min}L. Li, Q.W. Wang, S.Q. Shen, M. Li, \emph{Eur. Phys. Lett.} \textbf{114} (2016) 10007.
\bibitem{skew_prop}E. H. Lieb, \emph{Adv. Math.} \textbf{11} (1973) 267.
\bibitem{total_uncertainty}S. Luo, \emph{Theo. Math. Phys.} \textbf{143(2)} (2015) 681-688.
\bibitem{operator_varince} Y. Sun and N. Li, \emph{Quant. Inf. Process.} \textbf{20} (2021) 25. 
\bibitem{skew_prop1} C-S. Yu, \emph{Phys. Rev. A} \textbf{95} (2017) 042337.
\bibitem{skew_prop2}S. Luo and Q. Zhang, \emph{Phys. Rev. A} \textbf{69} (2004) 032106.
\end{thebibliography}
\end{document}